\documentclass{mcom-l}

\newtheorem{theorem}{Theorem}[section]
\newtheorem{lemma}[theorem]{Lemma}

\theoremstyle{definition}
\newtheorem{definition}[theorem]{Definition}
\newtheorem{example}[theorem]{Example}

\newtheorem{openproblem}[theorem]{Open Problem}

\theoremstyle{remark}

\numberwithin{equation}{section}

\usepackage[utf8]{inputenc}
\usepackage[T1]{fontenc}
\usepackage{amsmath,amssymb,amsfonts} 
\usepackage{amsthm}
\usepackage{cite}
\usepackage{url}
\usepackage[english]{babel}
\usepackage[mathscr]{eucal}
\usepackage{graphicx}
\usepackage{xcolor}
\usepackage{hyperref}
\hypersetup{
    colorlinks=true,
    linktoc=all,
    linkcolor=black,
}
\DeclareUnicodeCharacter{2212}{$-$}
\begin{document}

\title{On the Classification of Weierstrass Elliptic Curves over $\mathbb{Z}_n$}


\author{Param Parekh}
\address{}
\curraddr{Stony Brook University, 100 Nicolls Road, Stony Brook, 11794, New York, USA}
\email{param.parekh@stonybrook.edu}
\thanks{}

\author{Paavan Parekh}
\address{}
\curraddr{Stony Brook University, 100 Nicolls Road, Stony Brook, 11794, New York, USA }
\email{paavan.parekh@stonybrook.edu}
\thanks{}

\author{Sourav Deb}
\address{}
\curraddr{Unitedworld Institute of Technology, Karnavati University, Gandhinagar 382422, Gujarat, India} \email{sourav@karnavatiuniversity.edu.in}
\thanks{A premilinary version of the paper is available at https://arxiv.org/abs/2310.11768}

\author{Manish K Gupta}
\address{}
\curraddr{Kaushalya: The Skill University, Ahmedabad 382424
Gujarat, India}
\email{mankg@guptalab.org}
\thanks{}


\date{}

\dedicatory{}

\begin{abstract}
 The elliptic curves are beautiful mathematical object with several applications. They have been studied well over finite and infinite fields.  In this work, we study  Weierstrass elliptic curves over the finite ring $\mathbb{Z}_n$ through classification. Our study is supported with extensive computational data. We also present some conjectures.
 \end{abstract}

\maketitle
\section{Introduction}
The elliptic curves over finite/infinite fields have been studied extensively by many researchers \cite{silverman2009arithmetic, washington2008elliptic}. Elliptic curves over finite fields were first introduced by Neal Koblitz \cite{koblitz1987elliptic} and Victor S. Miller \cite{miller1985use} independently in $1985$. Classification of elliptic curves over a given field or ring $\mathbb{K}$ is useful in identifying isomorphism classes. It allows selecting a representative that may lead to a more efficient implementation of the group law. The classification of elliptic curves over finite fields for Weierstrass curves \cite{menezes1993elliptic, jeong2009isomorphism, schoof1987nonsingular}, and various alternate models of elliptic curves \cite{farashahi2011number, rezaeian2010number, feng2010elliptic, wu2011isomorphism} has been studied well. Further, the number of isomorphism classes of hyperelliptic curves over finite fields has also been of interest \cite{choie2004isomorphism, choie2002isomorphism, deng2006isomorphism, encinas2002isomorphism}. Elliptic curves over several finite rings have been studied in \cite{SalaTaufer+2024, A.Chillali, Hassib, tadmori2015elliptic, Tadmori2015CryptographyOT, Boulbot2016EllipticCO}. 

In this paper, we present a detailed algebraic classification of the Weierstrass elliptic curves defined over the finite ring $\mathbb{Z}_n$.
The comprehensive dataset obtained by the computational approach is also given. This work primarily addresses the general setting over the finite ring $\mathbb{Z}_n$. In certain cases, however, the inherent algebraic difficulties in finding roots of polynomials in $\mathbb{Z}_n$ necessitate framing the results under the closest applicable general assumptions, accompanied by appropriate explanations. Additionally, all point enumerations on the elliptic curves are performed excluding the point at infinity.

This work is organized as follows. Section \ref{pre} introduces basic definitions and notations of elliptic curves. Section \ref{clafq} summarizes key results from the literature on classifying reduced and generalized Weierstrass elliptic curves over finite fields.  Section \ref{distKth} discusses the distribution of $k^{th}$ power residues over $\mathbb{Z}_n$, which is helpful in deriving upper bound for the number of reduced Weierstrass equations over $\mathbb{Z}_{p^{m}}$. Section \ref{cla} presents the main contributions, extending classification results for generalized and reduced Weierstrass elliptic curves over $\mathbb{Z}_n$. Section \ref{com} provides computational data in tabular form to support the results of Section \ref{cla}, along with a dedicated HTML page containing the complete classification database \cite{ECCguptalab}. Finally, Section \ref{con} concludes the paper.

                                
\section{Preliminaries}
\label{pre}
In this section, we revisit the basics of the elliptic curves. The Weierstrass equation of an elliptic curve $E$ is described by the following polynomial. 
\begin{equation}
    \begin{split}
       \quad E: y^2+a_1 x y+a_3 y=x^3+a_2 x^2+a_4 x+a_6.
    \end{split}
 \label{nonhom_eq}   
\end{equation}
Here,  the point $(0,1,0)$ in $E$ is called the point at infinity, denoted by $\mathcal{O}$ and $a_1, a_2, a_3, a_4, a_6 \in \mathbb{K}$, (often a field)  then $E$ is said to be defined over $\mathbb{K}$.

The characterizing factors of $E$ are the following values along with the discriminant $\Delta$ and j-invariant $j(E)$ \cite{menezes1993elliptic}.
\begin{equation}
          \begin{aligned}
            b_2 & =a_1^2+4 a_2 \\
            b_4 & =2 a_4+a_1 a_3 \\
            b_6 & =a_3^2+4 a_6 \\
            b_8 & =a_1^2 a_6+4 a_2 a_6-a_1 a_3 a_4+a_2 a_3^2-a_4^2 \\
            c_4 & =b_2^2-24 b_4 \\
            c_6 & =-b_2^3 + 36 b_2 b_4 - 216 b_6\\
            \Delta & =-b_2^2 b_8-8 b_4^3-27 b_6^2+9 b_2 b_4 b_6\\
            j&(E) = c_4^3/\Delta \\  
            1728\,\Delta &= c_4^3 - c_6^2.
         \end{aligned}
\end{equation}
    
 Equation \ref{nonhom_eq} represents the non-singular generalized Weierstrass elliptic curve $E / \mathbb{K}$, $i.e.$ $\Delta \neq 0$.
 It is well known that the set of points of the elliptic curves forms an abelian group under the chord-and-tangent rule, where the point at infinity $\mathcal{O}$ is the identity element (Theorem 2.3, \cite{menezes1993elliptic}). 
 The natural isomorphism between two elliptic curves $E_1 / \mathbb{K}$ and $E_2 / \mathbb{K}$ can be determined if and only if $j(E_1)=j(E_2)$.
 Equivalently, two elliptic curves given by the generalized Weierstrass equations over $\mathbb{K}$  are isomorphic if one curve can be obtained from the other using the coordinate transformation $\tau : (x,y) \rightarrow (u^2x + r, u^3y + u^2sx + t),\; u \in \mathbb{K}^*$, $r,s,t \in \mathbb{K}$.

While the generalized Weierstrass equation can be used over any field with random characteristic, the reduced Weierstrass equation is crucial in the case of $char(\mathbb{K})\neq 2,3$, by selecting the transformation mapping as $(x,y) \xrightarrow{} (x, y - \frac{a_1}{2}x - \frac{a_3}{2})$ and $(x,y) \rightarrow (\frac{x-3b_2}{36}, \frac{y}{216})$  \cite{menezes1993elliptic} 
over $\mathbb{K}$ as,

\begin{equation}
        E  : y^2 = x^3 + ax + b,\ char(\mathbb{K}) \neq 2,3.
    \end{equation}

The j-invariant and discriminant  for the reduced Weierstrass equation can be obtained as:

 \begin{equation}
    \begin{aligned}
            \Delta & = -16(4a^3+27b^2)\\
            j(E) &= -1728\frac{4a^3}{\Delta}.\\
      \end{aligned}
\end{equation}

Isomorphism between two reduced Weierstrass elliptic curves can be determined if one curve can be obtained from the other using the coordinate transformation $\tau : (x,y) \rightarrow (u^2x, u^3y),\; u \in \mathbb{K}^*$. 

Considering an elliptic curve $E$ defined over $\mathbb{K}$, we denote the automorphism group of $E$ by $Aut(E)$, which consists of all the isomorphisms to itself and subsequently, $|Aut(E)|$ denotes the cardinality of the automorphism group of $E$. 
One can note that the order of $Aut(E)$ takes different possible values based on the coefficients of the curve $E$ defined over $\mathbb{K}$. 
The automorphism group $Aut(E)$ holds a crucial role in characterizing the elliptic curves in classes, where each class contains a primary curve or the class leader that further corresponds to the remaining class members through isomorphisms.

Throughout this paper, we adopt the following notation. Let \(A\) represent either a finite field \(\mathbb{F}_q\) or a finite ring \(R\). The number of non-singular generalized and reduced Weierstrass elliptic curves over \(A\) are denoted by \(N_g(A)\) and \(N_r(A)\), respectively, where subscript \(g\) stands for generalized and subscript \(r\) stands for reduced. The number of unique non-singular generalized and reduced Weierstrass elliptic curves isomorphic to the given curve \(E_k/A\) are denoted by \(N_g^{(k)}(A)\) and \(N_r^{(k)}(A)\), where \(k\) labels the isomorphism class represented by \(E_k\). Correspondingly, \(C_g(A)\) and \(C_r(A)\) denote the number of isomorphism classes for these curves.

\section{Classification of Weierstrass Elliptic Curves over $\mathbb{F}_q$}
\label{clafq}
In this section, we collect the fundamental results on the elliptic curves over finite fields that are relevant to our work.
The following results are well known. 

\subsection{Reduced Weierstrass elliptic curves over $\mathbb{F}_q$} 
\begin{theorem} \cite{menezes1993elliptic}
\label{rwt1}
The number of non-singular reduced Weierstrass elliptic curves over a finite field $\mathbb{F}_q$ is $N_r(\mathbb{F}_q) = q^2 - q$, where $char(\mathbb{F}_q) \neq 2,3$.
\end{theorem}   
\begin{theorem} \cite{menezes1993elliptic}
 The number of unique non-singular reduced Weierstrass elliptic curves isomorphic to the given curve $E_k/\mathbb{F}_q$ is $N_r^{(k)}(\mathbb{F}_q) = \frac{q-1}{|Aut(E)|}$, where $char(\mathbb{F}_q) \neq 2,3$.
\label{rwt2}
\end{theorem}                 

\begin{theorem} \cite{menezes1993elliptic} For $char(\mathbb{F}_q) \neq 2,3$, the number of isomorphism classes of reduced Weierstrass elliptic curves over $\mathbb{F}_q$,  is 
    \[
C_r(\mathbb{F}_q) = \left \{ \,
\begin{array}{lll}
2q+6 & when & q \equiv 1 \pmod {12}\\
2q+2 & when & q \equiv 5 \pmod {12} \\
2q+4 & when & q \equiv 7 \pmod {12} \\
2q & when & q \equiv 11 \pmod {12}.\\
\end{array}
\right.\\
\]
\label{irw}
\end{theorem}    


\subsection{Generalized Weierstrass Elliptic Curves over $\mathbb{F}_q$}

\begin{theorem}
    The number of non-singular generalized Weierstrass elliptic curves over $\mathbb{F}_q$ is $N_g(\mathbb{F}_q)=q^5-q^4$.
    \label{gwt1}
\end{theorem}

\begin{proof}
First, we will find the total number of singular generalized Weierstrass elliptic curves over $\mathbb{F}_q$. Let,
\[
f(x,y)
=
y^2 + a_1 x y + a_3 y
- x^3 - a_2 x^2 - a_4 x - a_6.
\]

The curve is singular if and only if there exists
$(x_0,y_0) \in \mathbb{F}_q^2$ such that
\[
f(x_0,y_0)=0, \qquad
\frac{\partial f}{\partial x}(x_0,y_0)=0, \qquad
\frac{\partial f}{\partial y}(x_0,y_0)=0.
\]
\[
\frac{\partial f}{\partial x}
=
a_1 y - 3x^2 - 2a_2 x - a_4,
\qquad
\frac{\partial f}{\partial y}
=
2y + a_1 x + a_3.
\]

If curve is singular ($\Delta = 0$), the curve has a unique singular point,
which is either a node or a cusp  \cite{silverman2009arithmetic}.
The singularity conditions at $(0,0)$ leads to,
\[
a_3 = 0, \qquad a_4 = 0, \qquad a_6 = 0.
\]
Thus every such singular curve is of the form
\begin{equation}\label{eq:normalized}
y^2 + a_1 x y = x^3 + a_2 x^2,
\end{equation}
with $a_1,a_2 \in \mathbb{F}_q$.
There are exactly $q^2$ such normalized singular equations.

We now show that the set of singular generalized Weierstrass equations over $\mathbb{F}_q$
is in bijection with $\mathbb{F}_q^4$. We can prove bijection with following:

Start with a normalized singular equation \eqref{eq:normalized},
determined by $(a_1,a_2) \in \mathbb{F}_q^2$.
Choose an arbitrary translation $(x_0,y_0) \in \mathbb{F}_q^2$.
Applying the translation $x \rightarrow x + x_0, y \rightarrow y + y_0$ yields a generalized Weierstrass equation
$$y^2+a_1' x y+a_3' y=x^3+a_2' x^2+a_4' x+a_6'$$
with coefficients
\[
\begin{aligned}
a_1' &= a_1, \\
a_3' &= 2y_0 + a_1 x_0, \\
a_2' &= a_2 + 3x_0, \\
a_4' &= 3x_0^2 + 2a_2 x_0 - a_1 y_0, \\
a_6' &= x_0^3 + a_2 x_0^2 - y_0^2 - a_1 x_0 y_0.
\end{aligned}
\]

Surjectivity follows since every singular curve has a unique
singular point and hence admits a unique normalization of the form
\eqref{eq:normalized}.
Injectivity follows because two distinct quadruples
$(a_1,a_2,x_0,y_0)$ produce distinct equations.
Thus the mapping is bijective.
Since there are $q^4$ choices of $(a_1,a_2,x_0,y_0)$, there are exactly
$q^4$ singular curves. The bijection relies solely on the uniqueness of the singular point, which hold in all characteristics, hence proved.
\end{proof}

\begin{theorem}
      The number of unique generalized Weierstrass elliptic curves isomorphic to given curve $E_k/\mathbb{F}_q$ is $N_g^{(k)}(\mathbb{F}_q) = \frac{q^4-q^3}{|Aut(E)|}$.   
\label{nec}
\end{theorem} 
\begin{proof}
    Considering the transformation function $\tau : (x,y) \rightarrow (u^2x + r, u^3y + u^2sx + t), u \in \; \mathbb{F}_q^*$ and Theorem \ref{rwt2}, the result follows directly.
\end{proof}
Waterhouse \cite{waterhouse1969abelian} (see also \cite{schoof1987nonsingular}) counted the number of isomorphism classes of the elliptic curve defined over finite field $\mathbb{F}_q$ that holds an extreme motivation behind this work.   
\begin{theorem} [Proposition $5.7$, \cite{schoof1987nonsingular}] 
      The number of isomorphism classes of generalized Weierstrass elliptic curves over $\mathbb{F}_q$ is $C_g(\mathbb{F}_q) = 2q + 3 + \left(\frac{-4}{q}\right) + 2\left(\frac{-3}{q}\right)$, where $\left(.\right)$ denotes the Jacobi symbol.  
\label{igw}      
\end{theorem} 

Before going into the analysis over finite rings $\mathbb{Z}_{n}$, in the following section, we present some of the number theoretic results on $k^{th}$ power residue over $\mathbb{Z}_n$, which helps us in finding upper bound on the number of reduced Weierstrass elliptic curves over $\mathbb{Z}_{p^m}$ for odd prime $p$.

\section{Distribution of $k^{th}$ Power Residue over $\mathbb{Z}_{n}$}
\label{distKth}

 Distribution of $k^{th}$ power residue over $\mathbb{Z}_{n}$ refers to determining the total number of $k^{th}$ power residues and their respective multiplicities over $\mathbb{Z}_{n}$.The problem of counting general power residues over $\mathbb{Z}_n$ has been studied in works such as \cite{stangl1996counting}, \cite{MatharConj} and \cite{serajcounting}. In this section, we analyze the multiplicity of power residues. Multiplicity of a $k^{\text{th}}$ power residue refers to the total number of distinct values of $x$ that satisfy the congruence $x^k \equiv a \pmod{n}$. We derive theorems for the multiplicity of $k^{\text{th}}$ power residues over odd prime power moduli. 

 Unless stated otherwise, throughout the following definitions, $n$ is assumed to be a natural number, i.e., $n \in \mathbb{N}$.
 \begin{definition}
An element $a \in \mathbb{Z}_n $ is called a unit if $\gcd(a, n) = 1 $; otherwise, $a$ is called a non-unit. Units are precisely the elements that are coprime to $n$.
\end{definition}

\begin{definition}
Euler’s totient function $\phi(n)$ is the number of unit elements in $ \mathbb{Z}_n$.
\end{definition}

\begin{definition}  \cite{rosen2011elementary}
An integer $g$ is called a primitive root modulo $n$ if for every integer $a$ that is coprime to $n$, there exists an integer $k$ such that $ g^k \equiv a \pmod{n} $. Here, $g$ is a generator of the multiplicative group of integers modulo $n$ and $k$ is called the index or discrete logarithm of $a$ to the base $g$ modulo $n$ denoted as $\operatorname{ind}_g a$. Note that, a primitive root exists if and only if $n$ is 1, 2, 4, $p^k$ or $2p^k$, where $p$ is an odd prime and $k > 0$.
\end{definition}

\begin{definition}
A quadratic residue ($QR$) modulo $n$ is an integer $a$ for which there exists an integer $x$ such that $x^2 \equiv a \pmod{n}$ (i.e., $a$ is a perfect square modulo $n$). If no such $x$ exists, then $a$ is called a quadratic nonresidue ($NQR$) modulo $n$. For instance, over the ring $\mathbb{Z}_5$, the elements $0$, $1$, and $4$ are $QRs$, since these are the values attained by squaring elements modulo $5$. The remaining elements of $\mathbb{Z}_5$ are thus $NQR$s.
\end{definition}

\begin{definition}
An integer $a$ is a $k^{th}$ power residue ($k \geq 2$) modulo $n$ if there exists an integer $x$ such that $x^k \equiv a \pmod{n}$.
\end{definition}

We employed computational methods to analyze the distribution of quadratic, cubic residues ($CR$s, $k=3$) and six-th power residues ($k = 6$) modulo $p^m$, where $p$ is odd prime. Computational data \cite{ECCguptalab} indicated that power residues exhibit distinct multiplicity patterns, allowing a distribution of $k^{th}$ power residues based on their multiplicity.  

For a given $u \in \mathbb{Z}_{p^m}$ and a positive integer $k$, we define the solution set $A^k_u(p^m) = \{x \in \mathbb{Z}_{p^m} \mid x^k \equiv u \pmod{p^m} \}$. The cardinality of the set $A^k_u(p^m)$, that is, the number of solutions to the congruence plays a fundamental role in understanding and classifying the $k^{th}$ power residues based on the increasing order of the multiplicity (excluding $u = 0$).

The subsequent example states the distribution of QRs over $\mathbb{Z}_{25}$ and $\mathbb{Z}_{3125}$, in detail and hence devises the framework.

\begin{example}
    \cite{ECCguptalab} Consider the ring $\mathbb{Z}_{25}$ and the following Tables \ref{tbl1} and \ref{tbl2}.
\begin{table}[h]
\centering 
\begin{tabular}{|c|c|c|c|c|}
\hline
$u$ & $A^2_u (25)$ & $|A^2_u (25)|$ \\  
\hline
$0$ & $\{0,5,10,15,20\}$ & $5$ \\
\hline
$1$ & $\{1,24\}$ & $2$ \\
\hline
$4$ & $\{2,23\}$ & $2$ \\ 
\hline
$6$ & $\{9,16\}$ & $2$ \\
\hline
$9$ & $\{3,22\}$ & $2$ \\
\hline
$11$ & $\{6,19\}$ & $2$ \\ 
\hline
$14$ & $\{8,17\}$ & $2$ \\
\hline
$16$ & $\{4,21\}$ & $2$ \\
\hline
$19$ & $\{12,13\}$ & $2$ \\
\hline
$21$ & $\{11,14\}$ & $2$ \\
\hline
$24$ & $\{7,18\}$ & $2$ \\
\hline

\end{tabular}
\caption{Analysis of the QRs over $\mathbb{Z}_{25}$.}
\label{tbl1}
\end{table}
Based on the multiplicities of the QRs (excluding $u = 0$), we present the class $[1]$ of the QRs over $\mathbb{Z}_{25}$. Similarly, the classification of the QRs (excluding $0$) over $\mathbb{Z}_{3125}$ is given in Table \ref{tbl3}.
\begin{table}[h]
\centering
\begin{tabular}{|c|c|c|}
\hline
Classes & multiplicity of $u$  & Total no. of  \\ 
& & $u$ with same multiplicity \\
\hline
 $[1]$ & $2$ & $10$       \\ \hline
\end{tabular}
\caption{Classification of QRs over $\mathbb{Z}_{25}$.}
 \label{tbl2} 
\end{table}

\begin{table}[h]
\centering
\begin{tabular}{|c|c|c|}
\hline
Classes & multiplicity of $u$  & Total no. of  \\ 
& & $u$ with same multiplicity \\
\hline
 $[1]$ & $2$ & $1250$       \\ \hline
 $[2]$ & $10$ & $50$       \\ \hline
 $[3]$ & $50$ & $2$       \\ \hline
\end{tabular}
\caption{Classification of QRs over $\mathbb{Z}_{3125}$.}
 \label{tbl3} 
\end{table}

\label{ex_class}
\end{example}
Equivalently, the classifications for the cubic and sixth-order residues are derived and presented in \cite{ECCguptalab}. 
Using computational data for $QRs$, $CRs$, and sixth-order residues ($k = 2, 3, 6$), we derive several results on $k^{th}$ power residues that appear to be largely unexplored in the existing literature, with the exception of the recent work by S. Samer\cite{serajcounting}.

\begin{lemma}

(\cite{euler_theoremata_1750, euler_theoremata_1761} Euler's criterion) Let $n \geq 2$ be an integer such that there exists a primitive root modulo $n$. Let $a \in \mathbb{Z}_{n}$ be an integer coprime to $n$, and $k \geq 2$ be an integer. Then the congruence
$$
x^k \equiv a \quad(\bmod \;n)
$$
has a solution $x$, meaning $a$ is $k^{th}$ power residue modulo $n$, if and only if
$$
a^{\ell} \equiv 1 \quad(\bmod\; n),
$$
where $\ell=\frac{\phi(n)}{\gcd(k, \phi(n))}$. The number of such $k^{th}$ power residues of $n$ is $\ell$ \cite{euler_criteria}. 
\label{col2}
\end{lemma}
We also observe the following Lemma (its proof is easy to see, so omitted), which may be known in the literature but we could not find it.  

\begin{lemma} Each unit $k^{th}$ power residue $a$ over $\mathbb{Z}_n$ (where $n \geq 2$ be an integer such that there exists a primitive root modulo $n$) is the $k^{th}$ power residue of exactly $\gcd(k, \phi(n))$ integers modulo $n$.
\label{occ_Kth_residue}
\end{lemma}



\begin{lemma}
    \cite{serajcounting} For odd prime $p$ and positive integer $m > k$, if $b$ is a $k^{th}$ power residue in $\mathbb{Z}_{p^{m-k}}$, then $bp^k$ will be a $k^{th}$ power residue in $\mathbb{Z}_{p^m}$.    
    \label{inn_qr}
\end{lemma}


 This shows that for the case of non-unit (except 0) $k^{th}$ power residue $a$ in $\mathbb{Z}_{p^m}$; $\gcd(a,p^m)=p^s$,  $s$ will be divisible by k. Using this fact, it is straightforward that for $m \leq k$, all the $k^{th}$ residues are units in $\mathbb{Z}_{p^m}$.  

\begin{theorem}
  For odd prime $p$ and a positive integer $m > k$, if $b$ is a $k^{\text{th}}$ power residue in $\mathbb{Z}_{p^{m-k}}$ with multiplicity $\omega$, then the element $bp^k$, which is a $k^{\text{th}}$ power residue in $\mathbb{Z}_{p^m}$, has multiplicity $p^{k-1} \omega$.

    \label{occnu}
\end{theorem}
\begin{proof}
  Using Lemma \ref{inn_qr}, let $bp^k$ be the non-unit $k^{th}$ power residue in $\mathbb{Z}_{p^m}$, 
    \begin{equation}
      y^k \equiv bp^k \pmod{p^m} 
      \label{ykth}
    \end{equation}
 where  $b$ is the $k^{th}$ power residue in $\mathbb{Z}_{p^{m-k}}$.
  \begin{equation}
      x^k \equiv b\pmod{p^{m-k}}.
      \label{xkth}
  \end{equation}
  Finding the multiplicity of non-unit $k^{th}$ power residues in $\mathbb{Z}_{p^{m}}$ is equivalent to finding the solution of equation \ref{ykth}. It can be obtained by multiplying $p^k$ both sides to the equation \ref{xkth} so we will have $y=xp$, which shows that different values of $y$ ($y'$) can be obtained by multiplying $p$ to the different values of $x$ ($x'$). 
   All corresponding solutions of $x$ can be written as  $x = x' + i p^{m-k}$, for $i \in \{0, 1, \dots, p^{k} - 1\}$ in $\mathbb{Z}_{p^{m}}$.
   Hence,
  \begin{align*}
      y' = (x' + ip^{m-k})p \\
      y' = x'p + ip^{m-k+1}.
  \end{align*}
  Values of $i \in \{0,1,...,p^{k-1}-1\}$ so that there are unique solutions in  $\mathbb{Z}_{p^{m}}$. Now if $\omega$ is the total number of possible $x$ (multiplicity of $b$) then the total number of possible $y$ (multiplicity of $bp^k$) will be $ p^{k-1}\omega$. Therefore result follows. 
  \end{proof}
We have the following Lemma for the multiplicity of $0$ as $k^{th}$ power residue in $\mathbb{Z}_{p^m}$.
\begin{lemma}
Taking into account $A^k_0 (p^m) = \{x\ |\ x^k \equiv 0 \pmod{p^m} \}$, we have $|A^k_0(p^m) | = p^{\lfloor \frac{(k-1)m}{k}\rfloor}$. Here $p$ is odd prime number.
 \label{lemx2eq0}
\end{lemma}
\begin{proof}
It is obvious to note that any solution $x$ to the equation $x^k \equiv 0 \pmod{ p^m}$ must be a non-unit in $\mathbb{Z}_{p^m}$. Thus, $x=rp^s$, where $r \in \mathbb{Z}_{p^m}$ is coprime to $p$ (i.e., $r$ is a unit). Now,
$$x^k \equiv r^k p^{sk} \equiv 0 \pmod{p^m}.$$
Since $r$ is a unit in $\mathbb{Z}_{p^m}$, $r^k$ is also a unit and does not contribute factors of $p$. Therefore, the condition reduces to $sk \geq m$, which is equivalent to $s \geq \frac{m}{k}$. Since $s$ must be an integer, this means $s \geq \lceil\frac{m}{k}\rceil$. The values of $x = r p^{s}$ in $\mathbb{Z}_{p^m}$ where $s \geq \lceil\frac{m}{k}\rceil$ are precisely the multiples of $p^{\lceil\frac{m}{k}\rceil}$ in $\mathbb{Z}_{p^m}$. There are total $\frac{p^m}{p^{\lceil \frac{m}{k}\rceil}}$ such multiples (including zero) and it can be written as below :
$$p^{m - \lceil \frac{m}{k} \rceil} = p^{\lfloor \frac{(k-1)m}{k}\rfloor}.$$
Hence proved.
\end{proof}
Note that for QRs and CRs taking $k=2$ and $k=3$ in above Lemma gives us:
\begin{itemize}
    \item[i)] $|A^2_0(p^m)| = p^{\lfloor\frac{m}{2}\rfloor}$
    \item[ii)] $|A^3_0(p^m)| = p^{\lfloor\frac{2m}{3}\rfloor}$.
\end{itemize}
\begin{theorem}
    Consider odd prime $p$ and the $i^{th}$ class of $k^{th}$ power residue $a \in \mathbb{Z}_{p^m}$ (excluding zero) where $gcd(a,p^m) = p^{k\cdot (i-1)}$.
    Then, the multiplicity of $k^{th}$ power residues in class $[i]$ of $\mathbb{Z}_{p^m}$, is 
    $gcd(k,\phi(p^m)).p^{(i-1)(k-1)}$ where $1\leq i \leq \left \lfloor \frac{m-1}{k} \right\rfloor+1$.
    \label{occkth}
\end{theorem}
\begin{proof}
    Using Lemma \ref{inn_qr}, we can see that there will be $\left \lfloor \frac{m-1}{k} \right\rfloor+1$ classes for $k^{th}$ power residues over $\mathbb{Z}_{p^m}$. Using Lemma \ref{occnu}, It is easy to see the recursion formula below in this case where $\Omega(i,p^m)$ represents the multiplicity of $k^{th}$ power residues of $i$-th class in $\mathbb{Z}_{p^m}$. Note that $\Omega(i,p^m)$ depends on the $({i-1})^{th}$ class in $\mathbb{Z}_{p^{m-k}}$ as per the definition of the class.
    $$\Omega(i,p^m) = p^{k-1}\cdot \Omega(i-1,p^{m-k})$$, 
    where $2 \leq i \leq \left \lfloor \frac{m-1}{k} \right\rfloor+1$. The base condition will be the multiplicity of $1$-st class of $k^{th}$ power residues (residues that are units), i.e., $\Omega(1,p^r)=gcd(k,\phi(p^r))\;r \in \mathbb{Z}\; (\because Lemma\;\ref{occ_Kth_residue})$. 
\begin{equation*}
    \begin{split}
        \Omega(i,p^m) &= p^{k-1}. \Omega(i-1,p^{m-k})\\
        &=  p^{k-1}.p^{k-1}. \Omega(i-2,p^{m-2k})\\
        &= p^{2(k-1)} . \Omega(i-2,p^{m-2k})\\
        &= p^{3(k-1)} . \Omega(i-3,p^{m-3k})\\
        & \vdots\\
        &= p^{(i-1)(k-1)}. \Omega(1,p^m)\\
        \Omega(i,p^m)&= gcd(k,\phi(p^m)).p^{(i-1)(k-1)}.
    \end{split}
    \end{equation*}
     Hence proved.
    \end{proof}
Using Theorem \ref{occkth}, one can derive the multiplicity of quadratic residues and cubic residues:
\begin{itemize}
    \item The multiplicity of $QR$s in class $[i]$ of $\mathbb{Z}_{p^m}$ is $2p^{i-1}$ where $1 \leq i \leq \left \lfloor \frac{m-1}{2} \right\rfloor+1$.
    \item For $p \equiv 2 \pmod 3$, the multiplicity of $CR$s  in class $[i]$ of $\mathbb{Z}_{p^m}$, is $p^{2(i-1)}$ and for $p \equiv 1 \pmod 3$, the multiplicity of $CR$s in class $[i]$ of $\mathbb{Z}_{p^m}$ is $3p^{2(i-1)}$ where $1\leq i \leq \left \lfloor \frac{m-1}{3} \right\rfloor+1$.
\end{itemize}
For example, the multiplicity of QRs in class [3] of $\mathbb{Z}_{3125}$ is $2\cdot5^{3-1} = 2\cdot5^{2} = 50$ which is same as presented at table \ref{tbl3}.

We have arrived at the following general result based on Lemma \ref{col2} and Lemma \ref{inn_qr} and by observing computational data for $k^{th}$ power residues, mainly $k=2$, $k=3$ and $k=6$. Total number of $k^{th}$ power residues over $\mathbb{Z}_{p^m}$ can be obtained by adding number of $k^{th}$ power residues in each class (starting from 0 in below lemma). This formula coincides with the formula given in Theorem 3.2 of \cite{serajcounting}.
\begin{lemma}
     The number of $k^{th}$ power residues over $\mathbb{Z}_{p^m}$ where $p$ is an odd prime number, is
    \begin{equation*}
       \sum_{i=0}^{\left \lfloor \frac{m-1}{k} \right\rfloor}\frac{\phi(p^{m-ki})}{gcd(k,\phi(p^{m-ki}))}+1.    
    \end{equation*}
    \label{conj3}   
\end{lemma}
For instance, 
     if $p \equiv 2 \pmod 3$, the number of sixth residues in $\mathbb{Z}_{p^m}$ is 
    
    \begin{equation*}
       (\sum_{i=0}^{\left \lfloor \frac{m-1}{6} \right\rfloor}\frac{1}{2}\phi(p^{m-6i}))+1    
    \end{equation*}
    and For $p \equiv 1 \pmod 3$, the number of sixth residues in $\mathbb{Z}_{p^m}$, is 

    \begin{equation*}
       (\sum\limits_{i=0}^{ \left \lfloor \frac{m-1}{6} \right\rfloor}\frac{1}{6}\phi(p^{m-6i}))+1.    
    \end{equation*}
    \label{6R_conj}
To identify whether a unit element of $\mathbb{Z}_{p^m}$ is a $k^{th}$ power residue or not, we can use Euler's criterion, but for a non-unit element of $\mathbb{Z}_{p^m}$, we have not seen a rule which can tell if it is a $k^{th}$ power residue or not. Hence, we are giving the following lemma. The $k^{th}$ power identification rule serves as a natural extension of the Euler criterion to $\mathbb{Z}_{p^m}$. Although the result is straightforward to establish, to the best of our knowledge, it has not been formally presented in the literature.
\begin{lemma}
    Euler's Criterion over $\mathbb{Z}_{p^m}, \forall a \in \mathbb{Z}_{p^m}$ including units and non units\\
    Let $gcd(a, p^m) = p^s$, $a = \gamma p^s \in \mathbb{Z}_{p^m}$ and $k \geq 2$ be an integer,  the congruence equation $x^k \equiv a \pmod {p^{m}}$ has solution if and only if below condition holds,
    \begin{equation*}
        s \equiv 0 \pmod k\;\; \text{and}\;\; \gamma^{\frac{\phi(p^{m-s})}{gcd(k,\phi(p^{m-s}))}} \equiv 1 \pmod {p^{m-s}}.
    \end{equation*}
    \label{egec}
\end{lemma}
\begin{proof}
For a non-unit element \( a = \gamma p^s \) in \( \mathbb{Z}_{p^m} \), it follows from Lemma \ref{inn_qr} that \( s = k\mu \), where \( \gamma \) is a unit \( k^{\text{th}} \) power residue in \( \mathbb{Z}_{p^{m-s}} \). Therefore, to determine whether \( a \) is a \( k^{\text{th}} \) power residue, it suffices to verify that \( s \) is divisible by \( k \) and that \( \gamma \) is a unit \( k^{\text{th}} \) power residue in \( \mathbb{Z}_{p^{m-s}} \). Hence, the claim is proved. For the unit element $a$ the values of $\gamma$ will be units in $\mathbb{Z}_{p^m}$ and $s = 0$.\\ 

\end{proof}

\section{Classification of Weierstrass Elliptic Curves over $\mathbb{Z}_n$}
\label{cla}









The elliptic curves over finite rings are studied similarly as presented in the finite field cases.
However, due to the generalized nature of rings, searching for the roots of a polynomial of degree $n$ is a fundamental problem.  
In this work, we are explicitly interested in devising a method to classify the elliptic curves over $\mathbb{Z}_n$ based on the underlying fundamental characteristics that closely align with the computational data.

The generalized Weierstrass equation of elliptic curves over $\mathbb{Z}_n$ can be represented as
\begin{equation}
 \begin{split}
        E  : y^2 + a_1xy + a_3y &=  x^3 + a_2x^2 + a_4x + a_6
 \end{split}
    \label{wcr}
    \end{equation}
    
Where $a_1, a_2,a_3,a_4, a_6 \in\;\mathbb{Z}_n$. The corresponding discriminant and $j$-invariant for the generalized Weierstrass equation can be obtained as, 

 \begin{equation*}
          \begin{aligned}
            \Delta & =-b_2^2 b_8-8 b_4^3-27 b_6^2+9 b_2 b_4 b_6\\
            j(E) &= c_4^3/\Delta \\
            c_4 & =b_2^2-24 b_4 \\
            b_2 & =a_1^2+4 a_2 \\
            b_4 & =2 a_4+a_1 a_3 \\
            b_6 & =a_3^2+4 a_6 \\
            b_8 & =a_1^2 a_6+4 a_2 a_6-a_1 a_3 a_4+a_2 a_3^2-a_4^2, 
      \end{aligned}
\end{equation*}

where discriminant $\Delta \in \mathbb{Z}_{n}^{*}$ (definition 4.3.1, \cite{schmaleelliptic}) and relating the transformation mapping as stated in the case of $\mathbb{F}_q$, we can obtain the  Weierstrass elliptic curves with the general transformation $\tau : (x,y) \rightarrow (u^2x + r, u^3y + u^2sx + t), u \in \; \mathbb{Z}_n^*, r,s,t \in \mathbb{Z}_n$.
Moreover, we are curious about the reduced equations of the elliptic curve over $\mathbb{Z}_n$, where $char(\mathbb{Z}_n) \nmid  2,3$ (section $3$, \cite{lenstra1986elliptic}).  


The reduced Weierstrass equation of elliptic curves over $\mathbb{Z}_n$ can be obtained by applying $(x,y) \xrightarrow{} (x, y - \frac{a_1}{2}x - \frac{a_3}{2})$ and $(x,y) \rightarrow (\frac{x-3b_2}{36}, \frac{y}{216})$ as admissible change of variables to Equation \ref{wcr}:
\begin{equation}
        E  : y^2 = x^3 + ax + b, \ \ a,b \in \mathbb{Z}_n\ \text{and}\ \gcd(6,n)=1.
    \end{equation}

Subsequently, the corresponding discriminant and $j$-invariant for the reduced Weierstrass equation will be in the form as,

 \begin{equation*}
    \begin{aligned}
            \Delta & = -16(4a^3+27b^2) \in \mathbb{Z}_{n}^{*}\\
            j(E) &= -1728\frac{4a^3}{\Delta}. \\
      \end{aligned}
\end{equation*}

For classifying these forms of the Weierstrass equation, the transformation function is $\tau : (x,y) \rightarrow (u^2x, u^3y),\; u \in \; \mathbb{Z}_n^*$.


  

We summarize the findings of this work in the following theorems. 
\subsection{Reduced Weierstrass Elliptic Curves over $\mathbb{Z}_n$}

Computational data through which we have arrived at these results are given in \cite{ECCguptalab}. 
Here, we set the central focus on devising the number of non-singular reduced Weierstrass elliptic curves over $\mathbb{Z}_n$.
Explicit proofs for $n=p^m$ and $n=p_1^{e_1}p_2^{e_2}\ldots p_k^{e_k}$ are outlined in Theorem \ref{nrwr} and Theorem \ref{nrc}, respectively. 
Upper bound ($N''_r(\mathbb{Z}_n)$) for $N_r(\mathbb{Z}_n)$, $n$ = $p^m$ is described in Theorem \ref{col5}.

 \begin{theorem}
  The number of reduced Weierstrass elliptic curves isomorphic to given curve $E_k/\mathbb{Z}_n$ is $N_r^{(k)}(\mathbb{Z}_n) = \frac{\phi(n)}{|Aut(E)|}$.
  \label{rwr}
 \end{theorem}           

 \begin{proof}
Considering the transformation function $\tau : (x,y) \rightarrow (u^2x, u^3y), u \in \; \mathbb{Z}_n^*$ and Theorem \ref{rwt2}, the result follows directly.
 \end{proof}

\begin{theorem}
      For $gcd(p^m,6)=1$, the number of non-singular reduced Weierstrass elliptic curves over $\mathbb{Z}_{p^m}$ is $N_r(\mathbb{Z}_{p^m}) = \phi(p^{2m})$.
    
    \label{nrwr}
\end{theorem}

\begin{proof}
Let $E/\mathbb{Z}_{p^m}$ be a reduced Weierstrass elliptic curve given by
\[
E : y^2 = x^3 + ax + b,
\]
where $a,b \in \mathbb{Z}_{p^m}$ and $\gcd(p^m,6)=1$. 
The corresponding discriminant is
\[
\Delta = -16(4a^3 + 27b^2).
\]
Here, the constant $-16$ is a unit in $\mathbb{Z}_{p^m}$, and therefore
$\Delta$ is a unit in $\mathbb{Z}_{p^m}$ if and only if $4a^3 + 27b^2$ is a unit in $\mathbb{Z}_{p^m}$.

There are $p^{2m}$ ordered pairs $(a,b) \in \mathbb{Z}_{p^m}^2$. Hence, it suffices to count those pairs for which $4a^3 + 27b^2$ is a non-unit. If $4a^3 + 27b^2$ is a non-unit in $\mathbb{Z}_{p^m}$, then by Theorem~\ref{rwt1}, the congruence
\[
4a^3 + 27b^2 \equiv 0 \pmod{p}, 
\]
admits exactly $p$ solutions $(a_0,b_0) \in \mathbb{Z}_p^2$. Any lift of $(a_0,b_0)$ to $\mathbb{Z}_{p^m}$ is uniquely
expressible as
\[
a = a_0 + pA, \qquad b = b_0 + pB,
\]
with $A,B \in \{0,1,\dots,p^{m-1}-1\}$. Substituting into $4a^3 + 27b^2$ and expanding,
all terms except $4a_0^3 + 27b_0^2$ contain a factor of $p$. Thus,
\[
4a^3 + 27b^2 \equiv 4a_0^3 + 27b_0^2 \equiv 0 \pmod{p}
\]
for every choice of $A$ and $B$, and hence every such lift yields a non-unit in
$\mathbb{Z}_{p^m}$.

For each solution modulo $p$, there are $p^{m-1}$ choices for $A$ and $p^{m-1}$ choices
for $B$, giving $p^{2(m-1)}$ lifts. Since there are $p$ solutions modulo $p$, the total
number of pairs $(a,b) \in \mathbb{Z}_{p^m}^2$ for which the discriminant is a non-unit is
$p \cdot p^{2(m-1)} = p^{2m-1}$. Subtracting this from the total number of pairs yields
\[
N_r(\mathbb{Z}_{p^m}) = p^{2m} - p^{2m-1} = p^{2m-1}(p-1) = \phi(p^{2m}).
\]
This completes the proof.
\end{proof}
As an immediate consequence, the following result reflects the extension of Theorem \ref{nrwr} on the odd composite number $n=p_{1}^{e_1} p_{2}^{e_2} \ldots p_{l}^{e_l}$ such that $gcd(n,6)=1$.
 
\begin{theorem}
    For the odd composite number $n=p_{1}^{e_1} p_{2}^{e_2} \ldots p_{l}^{e_l}$ such that $gcd(n,6)=1$, $N_r(\mathbb{Z}_n) = \prod\limits_{i=1}^{l}N_r(\mathbb{Z}_{p_i^{e_i}}) = \phi(n^2)$. 
    \label{nrc}
\end{theorem}

\begin{proof}

The following expression can be established considering the class leaders $\mathbb{E}_k$s over $\mathbb{Z}_n$. 

 \begin{equation}
      \begin{split}
          N_{R}(\mathbb{Z}_n) &= \sum_{\mathbb{E}_k} N^{(k)}_r(\mathbb{Z}_{n}),  
      \end{split}
      \label{ncraute4}
  \end{equation}
where $N^{(k)}_R(\mathbb{Z}_{n})$ represents the total number of reduced Weierstrass elliptic curves over $\mathbb{Z}_{n}$ which are isomorphic to the class leader $\mathbb{E}_k$. 

Using Theorem \ref{rwr} and the following result, we can derive the expression for $N^{(k)}_R(\mathbb{Z}_{n})$ as follows.

    \begin{equation}
    \begin{split}
        |Aut(\mathbb{E}_k)|\; \text{over}\; \mathbb{Z}_n  = \prod\limits_{i=1}^{l} |Aut(\mathbb{E}_{k_i})|\; \text{over}\; \mathbb{Z}_{p_i^{e_i}}
        \text{ (Proposition 2.13, \cite{kayal2006complexity})}
    \end{split}
    \end{equation}
    

\begin{equation*}
    \begin{split}
        N^{(k)}_r(\mathbb{Z}_{n}) &= \frac{\phi(n)}{|Aut(E_k)|}\\
        &= \frac{\prod\limits_{i=1}^{l}\phi(p_i^{e_i})}{\prod\limits_{i=1}^{l}|Aut({E_k}_i)|}\\
        &= \prod_{i=1}^{l} \frac{\phi(p_i^{e_i})}{|Aut({E_k}_i)|}\\
      N^{(k)}_r(\mathbb{Z}_{n})  &= \prod_{i=1}^{l} N^{(k_i)}_{R}(\mathbb{Z}_{p_i^{e_i}}).
    \end{split}
\end{equation*}

Since each elliptic curve over $\mathbb{Z}_n$ reduces to a curve modulo each $p_i^{e_i}$, every class leader $\mathbb{E}_k$ over $\mathbb{Z}_n$ corresponds to a tuple $(\mathbb{E}_{k_1},\ldots,\mathbb{E}_{k_l})$ of class leaders over the $\mathbb{Z}_{p_i^{e_i}}$.

\begin{equation}
    \begin{split}
N_{R}(\mathbb{Z}_n) 
    &=  \sum_{\mathbb{E}_k} N^{(k)}_r(\mathbb{Z}_{n}) \\
    &=  \sum_{(k_1,\ldots,k_l)} \prod_{i=1}^{l} 
        N^{(k_i)}_{R}(\mathbb{Z}_{p_i^{e_i}}) 
        \quad\text{(each $\mathbb{E}_k$ corresponds to a tuple $(\mathbb{E}_{k_1},\ldots,\mathbb{E}_{k_l})$ )} \\
    &=  \prod_{i=1}^{l} \sum_{\mathbb{E}_{k_i}} 
        N^{(k_i)}_{R}(\mathbb{Z}_{p_i^{e_i}}) \\
N_{r}(\mathbb{Z}_n)    &=  \prod_{i=1}^{l} N_{r}(\mathbb{Z}_{p_i^{e_i}}).
    \end{split}
\label{ncraute5}
\end{equation}
By multiplicativity of $\phi$,
\begin{equation}
    \begin{split}
        N_{r}(\mathbb{Z}_n) &= \prod_{i=1}^{l} N_{r}(\mathbb{Z}_{p_i^{e_i}})\\
        &= \prod_{i=1}^{l} \phi({p_i^{2e_i}})\\
       N_{r}(\mathbb{Z}_n) &= \phi(n^2).
    \end{split}
\end{equation}
\end{proof}

We give upper bound on $N_r(\mathbb{Z}_{p^m})$ which uses results from Section 4.
Considering $\Delta^{i}(n)$ as the number of solutions to $\Delta\equiv i \pmod n$, we state the following claim that complies with the computational data attained separately.   

\begin{theorem}
    The upper bound $N''_r(\mathbb{Z}_{p^m})$ on the number of non-singular reduced Weierstrass elliptic curves over $\mathbb{Z}_{p^m}$ 
    is given by  
    \begin{equation}
        N_r(\mathbb{Z}_{p^m}) \leq N''_r(\mathbb{Z}_{p^m})  = (p^{2m}-\Delta^{0}(p^m)), 
    \end{equation}
    where $p$ is odd prime, $\gcd(p^m,6)=1$.
      \label{col5}
\end{theorem}
\begin{proof}

$\Delta \equiv 0 \pmod{p^m}$ leads to simple equation
   \begin{equation}
      \begin{split}
            s_1^3 + s_2^2 &\equiv 0 \pmod{p^m}  \\
          \implies  s_1^3 &\equiv - s_2^2 \pmod{p^m}\\
          \implies s_1^3 &\equiv (p^m - 1)\cdot s_2^2  \pmod{p^m}, 
          \label{crqrpm}
      \end{split}
      \end{equation}
      where $(3^{-1}a,2^{-1}b)=(s_1,s_2)$ since 27 and 4 are $CR$ and $QR$ over $\mathbb{Z}_{p^m}$ respectively and $2,3$ are units in $\mathbb{Z}_{p^m}$. We want to count cubic residues that follows Equation \ref{crqrpm}.  Each class $i$ of cubic residues and class $x$ of quadratic residues can be represented as $\gamma \cdot p^{3i}$ and $\gamma' \cdot p^{2x}$ respectively. Cubic residues those are in even class ($i \equiv 0 \pmod{2}$) can be represented in form of Equation \ref{crqrpm}. To get the number of solutions of Equation \ref{crqrpm},  the number of cubic residues in class $i , i \equiv 0 \pmod{2}$ is multiplied with multiplicity of those CRs and QRs in class $i$ and $\frac{3i}{2}$ respectively. Note that the number of those cubic residues in class $i$ is exactly half of the number of cubic residues in class $i$ where $i \equiv 0\pmod{2}$. This is because Equation \ref{crqrpm} covers the case of $CR = QR$ or $CR = NQR$ depending on whether $p \equiv 1 \pmod4$ or $p \equiv 3 \pmod4$. \\
      Sixth power residues in class $\frac{i}{2}$ are exactly the elements which are cubic residues (as well quadratic residues) of class $i, i \equiv 0 \pmod{2}$. Hence, the number of solution to $CR=QR$ in class $i$ is exactly the number of sixth power residues of class $\frac{i}{2}$. As per Lemma \ref{conj3}, it follows.

      If in class $i$ number of solution to $CR=QR$ is half of the number of $CR$, then it immediately follows that $CR=NQR$ is also exactly half of the number of $CR$. 

\begin{equation*}
\begin{split}
    \Delta^{0}(p^m) &= \sum_{\substack{i=0 \\ i\equiv0 \pmod{2}}}^{ \left \lfloor \frac{m-1}{3} \right\rfloor}\frac{\phi(p^{m-3i})}{2\cdot \gcd(3,\phi(p^{m-3i}))}\cdot \gcd(3,\phi(p^m))p^{2i} \cdot \gcd(2,\phi(p^m))p^{\frac{3i}{2}}\\
    & + p^{\left \lfloor \frac{m}{2} \right\rfloor + \left \lfloor \frac{2m}{3} \right\rfloor}\\
      &= \sum_{\substack{i=0 \\ i\equiv0 \pmod{2}}}^{ \left \lfloor \frac{m-1}{3} \right\rfloor}\frac{\phi(p^{m-3i})}{\gcd(6,\phi(p^{m-3i}))}\cdot \gcd(6,\phi(p^m))p^{\frac{7i}{2}}\\
    & + p^{\left \lfloor \frac{m}{2} \right\rfloor + \left \lfloor \frac{2m}{3} \right\rfloor}.\\
\end{split}
\end{equation*}

\end{proof}
The supporting data is shown in Table \ref{NSredLB}. 

\begin{table}[ht]
\centering
\begin{tabular}{|c|c|c|}
\hline
  $p^m$   &   Actual value ($N_r(\mathbb{Z}_{p^m})$) & Upper bound ($N^{''}_r(\mathbb{Z}_{p^m})$)                       \\ \hline
$5^5$                              & $7812500$ &9760000\\ \hline
$5^6$                            & $195312500$ & 244050000\\ \hline
$5^7$                          & $4882812500$ & $6103062500$                  \\ \hline
$7^5$                           & $242121642$         & 282444036\\ \hline
$7^6$                           & 11863960458& 13840362816\\ \hline
$7^7$                         & 581334062442& 678216602154\\ \hline
\end{tabular}
\vspace{5pt}
\caption{Computational results on actual value($N_r(\mathbb{Z}_{p^m})$), and upper bound($N^{''}_r(\mathbb{Z}_{p^m})$) on the number of non-singular reduced Weierstrass equation over $\mathbb{Z}_{p^m}$. }
\label{NSredLB}
\end{table}

Here we have subtracted the exact number of pairs $(a,b)$ which gives solution to $\Delta^{0}(p^m)$. Solutions to $\Delta^{v}(p^m)$ where $v$ is other non-units, are not subtracted in this case and thus giving upper bound to the number of non-singular Weierstrass elliptic curves over $\mathbb{Z}_{p^m}$. Subtracting solutions of $\Delta^{v}(p^m)$ for all non-units will give you the actual value to the problem.


    



Finally, the subsequent results (Theorem \ref{irwr} and \ref{thrm10}) gives an exact formula to count the isomorphism classes of reduced Weierstrass elliptic curves over $\mathbb{Z}_{p^m}$ and $\mathbb{Z}_n$, where $n$ is a composite number.
 
 \begin{theorem}
  The number of isomorphism classes of reduced Weierstrass elliptic curves over the finite ring $\mathbb{Z}_{p^m},\ \gcd(p,6)=1$ will be,
  
    \[
C_r(\mathbb{Z}_{p^m}) = \left \{ \,
\begin{array}{lll}
2p^{m}+6 & when & p \equiv 1 \pmod{12}\\
2p^{m}+2 & when & p \equiv 5 \pmod{12} \\
2p^{m}+4 & when & p \equiv 7 \pmod{12} \\
2p^{m} & when & p \equiv 11 \pmod{12}.\\
\end{array}
\right.\\
\]
 \label{irwr}
 \end{theorem}

  \begin{proof}
     Consider the isomorphic reduced Weierstrass elliptic curves $E_1/\mathbb{Z}_{p^m} : y^2=x^3+ax+b$ and $E_2/\mathbb{Z}_{p^m} : y^2=x^3+\bar{a}x+\bar{b}$ along with the corresponding relations $\bar{a}=u^{-4}a$ and $\bar{b}=u^{-6}b$, $u \in {\mathbb{Z}_{p^m}^*}$ \cite{lenstra1986elliptic}.
     An analogous approach to Theorem \ref{irw} leads to the fact that the number of isomorphism classes of elliptic curves over $\mathbb{Z}_{p^m}$, corresponds to the $|Aut(E_1)|$ and we are left with three cases: {\bf a)} for $a \neq 0\ \text{and}\ b \neq 0\ (j(E_1) \neq 0,1728$) $|Aut(E_1)|=2$, {\bf b)} for $a=0\ \text{and}\ b \neq 0\ (j(E_1)=0$) either $|Aut(E_1)|=6$ or $|Aut(E_1)|=2$, and {\bf c)} for $a \neq 0\ \text{and}\ b=0\ (j(E_1)=1728$) either $|Aut(E_1)|=4$ or $|Aut(E_1)|=2$.  
     

    Subsequently, considering Theorem \ref{nrwr} and Theorem \ref{rwr}, we can immediately write the following equation,
    \begin{equation}
    \sum_{E_k}\frac{\phi(n)}{|Aut(E_k)|} = \phi(n^2), 
    \label{eqrc}
    \end{equation}
    where the summation is taken over the set of isomorphism class representatives of the elliptic curves defined over $\mathbb{Z}_{p^m}$.
    Since, $\gcd(p,6)=1$ we have $p \equiv 1,5,7,11\ \pmod{12}$.
    
    For the sake of simplicity, we consider the case when $p \equiv 1\ \pmod{12}$, and the rest of the cases will follow similarly.  
    
    The elementary number theoretic approach ensures $p \equiv 1\ \pmod{l}$, where $2 \leq l \leq 4$.
    So, we are left with the following sub-cases.

    
    \begin{itemize}
    
    \item[i)] $p \equiv 1 \pmod 2$ and hence $p^m \equiv 1 \pmod 2$, $i.e.$, $p = 2i+1$\ \text{and}\ $p^{m} = 2j+1,\ i, j \in \mathbb{N}$.
    A simple calculation shows that $2 \mid \phi(p^m)$, and so using Cauchy's theorem, we can ensure that there exists an element of order $2$ in $\mathbb{Z}_{p^{m}}^{*}$.   
    
        
     \item[ii)] $p \equiv 1 \pmod 3$ and hence $p^m \equiv 1 \pmod 3$, $i.e.$, $p = 3i+1$\ \text{and}\ $p^{m} = 3j+1,\ i, j \in \mathbb{N}$.
    Since $3 \mid \phi(p^m)$, using Cauchy's theorem, we ensure that there exists an element of order $3$ in $\mathbb{Z}_{p^{m}}^{*}$.

        
        
    \item[iii)] $p \equiv 1 \pmod 4$ and hence $p^m \equiv 1 \pmod 4$, $i.e.$, $p = 4i+1$\ \text{and}\ $p^{m} = 4j+1,\ i, j \in \mathbb{N}$.
    As $4 \mid \phi(p^m)$, using Sylow's first theorem, $\mathbb{Z}_{p^{m}}^{*}$ contains an element of order $4$.

        
        
    \end{itemize} 

  Therefore, from the above three cases, it can be inferred that the possible values of $|Aut(E_1)|$ are $2$, $4$, and $6$. 
  
  By fixing $|Aut(E)|=6$, we get $\Delta = -16(27b^2)$, since $a=0$. 
  Let $k_1$ be the number of isomorphism classes of reduced Weierstrass elliptic curves having the order of the respective automorphism group as $6$. 
  The possibilities of $b$ to be a unit will be, $\phi(p^m)$, and therefore we obtain the following relation

   \begin{equation*}
     \begin{aligned}
      (\frac{\phi(p^m)}{6})k_1 &= \phi(p^m) \\
      \implies k_1 &= 6. \\
     \end{aligned}
   \end{equation*}

 Similarly, if $k_2$ represents the number of isomorphism classes of reduced Weierstrass elliptic curves having the order of the respective automorphism group as $4$, then $k_2=4$. 
 Consider $k_3$ represents the number of isomorphism classes of reduced Weierstrass elliptic curves having the order of the respective automorphism group as $2$.
 Now we have $\phi(p^{2m})$ non-singular curves in total over $\mathbb{Z}_{p^m}$ (using Theorem \ref{nrc}), out of which there are $2\phi(p^m)$ curves that are having either $|Aut(E_k)|=4$ or $|Aut(E_k)|=6$. 
    So the total number of non-singular curves with $|Aut(E_k)|=2$ will be $\phi(p^{2m})- 2\phi(p^m)= \phi(p^m)(p^m - 2)$ and hence the number of isomorphism classes will be $2p^m-4$, using Equation \ref{eqrc}.

 Finally, if $p \equiv 1 \pmod{12}$, then there are $k_1 + k_2 + k_3 = 6 + 4 + 2p^m -4 = 2p^m + 6$ isomorphism classes of elliptic curves over $\mathbb{Z}_{p^m}$.
 Using an analogous approach, one can derive the different isomorphism classes for different relations for $p$, and this concludes the proof.
 \end{proof}
As the immediate consequence, we obtain the following result on the multiplicative form of the number of isomorphic classes given in Theorem \ref{irwr}.  

\begin{theorem}
For the odd composite integer $n=p_1^{e_1} p_2^{e_2} \ldots p_l^{e_l}$, $C_{r}(\mathbb{Z}_n) = \prod\limits_{i=1}^{l} C_r(\mathbb{Z}_{p_i^{e_i}})$. 
 \label{thrm10}
\end{theorem}

\begin{proof}
The number of isomorphism classes of reduced Weierstrass elliptic curves over the finite ring $\mathbb{Z}_{p_i} (1 \leq i \leq l)$, as per Theorem \ref{irw}, we have,
        \[
        C_r(\mathbb{Z}_{p_i}) = \left \{ \,
        \begin{array}{lll}
        2p_i+6 & when & p_i \equiv 1 \pmod {12}\\
        2p_i+2 & when & p_i \equiv 5 \pmod{12} \\
        2p_i+4 & when & p_i \equiv 7 \pmod{12} \\
        2p_i   & when & p_i \equiv 11 \pmod {12}.\\
        \end{array}
        \right.\\
        \]
        
Considering all the possible isomorphism classes over the class representatives, we obtain,
\begin{equation}
      \begin{split}
          C_{r}(\mathbb{Z}_n) = \sum_{\mathbb{E}_k} C^{(k)}_{r}(\mathbb{Z}_{n}), 
      \end{split}
      \label{craute1}
  \end{equation}

where $C^{(k)}_r(\mathbb{Z}_{n}) = 1$ represent the isomorphic class with the class leaders $\mathbb{E}_k$ in reduced form over $\mathbb{Z}_{n}$.
  
 Now Equation \ref{eqrc} leads to,  
  \begin{equation}
      C^{(k)}_{r}(\mathbb{Z}_n) = \frac{N^{(k)}_{r}(\mathbb{Z}_n).|Aut(\mathbb{E}_k)|}{\phi(n)}.
      \label{craute2}
  \end{equation}

For the odd composite integer $n=p_1p_2\ldots p_l$, we have the decomposition form $E(\mathbb{Z}_n) \cong E(\mathbb{Z}_{p_1}) \oplus E(\mathbb{Z}_{p_2}) \oplus \ldots \oplus E(\mathbb{Z}_{p_l})$ \cite{10.5555/1388394}.
This leads to the following results,
\begin{itemize}
    \item[a)]  $|Aut(\mathbb{E}_k)|\; \text{over}\; \mathbb{Z}_n  = \prod\limits_{i=1}^{l} |Aut(\mathbb{E}_{k_i})|\; \text{over}\; \mathbb{Z}_{p_i^{e_i}}$ (\cite{kayal2006complexity})\\
    \item[b)] $N^{(k)}_{R}(\mathbb{Z}_n) = \prod\limits_{i=1}^{l} N^{(k_i)}_{R}(\mathbb{Z}_{p_i^{e_i}})$    
\end{itemize}
Applying above results to Equation \ref{craute2} we will get,
  \begin{equation*}
      \begin{split}
      C^{(k)}_{r}(\mathbb{Z}_n) &= \frac{\prod\limits_{i=1}^{l} N^{(k_i)}_{r}(\mathbb{Z}_{p_i^{e_i}}). \prod\limits_{i=1}^{l} |Aut(\mathbb{E}_{k_i})|}{\prod\limits_{i=1}^{l}\phi(p_i^{e_i})} \\
      &= \prod\limits_{i=1}^{l} \frac{ N^{(k_i)}_{r}(\mathbb{Z}_{p_i^{e_i}}). |Aut(\mathbb{E}_{k_i})|}{\phi(p_i^{e_i})}\\
      &= \prod\limits_{i=1}^{l} C^{(k_i)}_{r}(\mathbb{Z}_{p_i^{e_i}}).
  \end{split}  
  \label{craute3}
  \end{equation*}
Therefore, from Equation \ref{craute1}, we have,
 \begin{equation}
      \begin{split}
          C_{r}(\mathbb{Z}_n) &= \sum_{\mathbb{E}_k}C^{(k)}_{r}(\mathbb{Z}_{n}) \\
          &= \sum_{\mathbb{E}_k} \prod\limits_{i=1}^{l} C^{(k_i)}_{r}(\mathbb{Z}_{p_i^{e_i}}).
      \end{split}
      \label{craute4}
  \end{equation}
  
Finally, considering the decomposition of each $\mathbb{E}_k$ over $\mathbb{Z}_{p_i}$ and the nature of the isomorphism, we obtain the desired result.
 \begin{equation}
      \begin{split}
           C_{r}(\mathbb{Z}_n) &= \prod\limits_{i=1}^{l} \sum_{\mathbb{E}_{k_i}} C^{(k_i)}_{r}(\mathbb{Z}_{p_i^{e_i}})\\
                               &= \prod\limits_{i=1}^{l} C^{}_{r}(\mathbb{Z}_{p_i^{e_i}}).
      \end{split}
      \label{craute5}
  \end{equation}
\end{proof}

\subsection{Generalized Weierstrass Elliptic Curves over $\mathbb{Z}_n$}
The computational data demonstrates that the number of the generalized Weierstrass elliptic curves over $\mathbb{Z}_n$ should be $\phi(n^5)$ \cite{ECCguptalab}.



\begin{theorem}
    The number of non-singular generalized Weierstrass elliptic curves over $\mathbb{Z}_{p^m}$, $N_g(\mathbb{Z}_{p^m})$ is $\phi(p^{5m})$.
\label{gwr} 
\end{theorem}

\begin{proof}
Using similar argument given in Theorem \ref{nrwr}, for each solution to $\Delta \equiv 0$ modulo $p$, there are $p^{(m-1)}$ choices for $a_i = a_i^{(0)} + p A_i, \;
\text{where} \; A_i \in \{0,1,\dots,p^{m-1}-1\}, a_i^{(0)} \text{is solution to $\Delta \equiv 0$ modulo $p$ } \;, 
\; i \in \{1,2,3,4,6\}$ giving $p^{5(m-1)}$ lifts. Since there are $p^{4}$ solutions modulo p, the total number of pairs $(a_i), i \in \{ 1,2,3,4,6\} \in \mathbb{Z}^5_{p^m}$ for which the discriminant is a non-unit is $p^4 \cdot p^{5(m-1)} = p^{5m-1}$.
Subtracting this from the total number of pairs yields
\begin{equation*}
    N_g(\mathbb{Z}_{p^m})
= 
p^{5m} - p^{5m-1}
=
p^{5m-1}(p-1)
=
\phi(p^{5m}).
\end{equation*}
This completes the proof.
\end{proof}

    \begin{theorem}
    For the composite number $n=p_{1}^{e_1} p_{2}^{e_2} \ldots p_{l}^{e_l}$, $N_g(\mathbb{Z}_n) = \prod\limits_{i=1}^{l}N_g(\mathbb{Z}_{p_i^{e_i}}) = \phi(n^5)$. 
    \label{ngc2}
\end{theorem}
\begin{proof}
    A similar argument to Theorem \ref{nrc} completes the proof.
\end{proof}

Corresponding computational validation is given in Table \ref{NSgen}.
 \begin{theorem}
    The number of unique generalized Weierstrass elliptic curves isomorphic to the given curve $E_k/\mathbb{Z}_n$ is $N_g^{(k)}(\mathbb{Z}_n) = \ \ \frac {\phi(n^4)}{|Aut(E)|}$.  
    \label{ugwr}
 \end{theorem}     

 \begin{proof}
 Considering the transformation function $\tau : (x,y) \rightarrow (u^2x + r, u^3y + u^2sx + t),\ u \in  \mathbb{Z}_n^*$, and $r,s,t \in \mathbb{Z}_n$, total number of elliptic curves isomorphic to given curve $E/\mathbb{Z}_n$ are $n^3\phi(n) = \phi(n^4)$ (since the possibilities of $u,r,s$ and $t$ are $\phi(n), n, n, n$ respectively in $\mathbb{Z}_n$).
 We obtain the desired result using a similar argument given in Theorem \ref{rwt2}.
 \end{proof}   
 \begin{openproblem}
     Compact formula as seen in Theorem \ref{igw} for the number of isomorphism classes of Weierstrass elliptic curves over $\mathbb{Z}_n$, i.e, $C_g(\mathbb{Z}_n)$ is unknown. 
    \label{igwr}
 \end{openproblem}
   
 It can be split into the following cases:
   \begin{enumerate}
       \item $n = p$ where $p$ is prime
       \item $n = p^m$ where gcd($p^m$,6) = 1
       \item $n = 2^m$
       \item $n = 3^m$
       \item $n$ is composite integer ($n=p_1^{e_1} p_2^{e_2} \ldots p_l^{e_l}$) where gcd($n$,6) = 1, $l >= 2$
       \item $n$ is composite integer ($n=p_1^{e_1} p_2^{e_2} \ldots p_l^{e_l}$) where gcd($n$,6) $\neq$ 1, $l >= 2$
   \end{enumerate}

For $\gcd(n,6)=1$, every non-singular generalized Weierstrass elliptic curve over $\mathbb{Z}_n$ is $\mathbb{Z}_n$-isomorphic to a reduced Weierstrass elliptic curve (see Section~\ref{pre}); hence the reduction induces a bijection on isomorphism classes, and therefore $C_g(\mathbb{Z}_n)=C_r(\mathbb{Z}_n).$ Considering the same, for the cases $1$, $2$ and $5$, $C_g (\mathbb{Z}_n)$ is given in Theorem \ref{irw}, Theorem \ref{irwr} and Theorem \ref{thrm10} respectively. 
 All the remaining cases are open problems.

The successive section presents a rigorous analysis of the elliptic curves over $\mathbb{Z}_n$ based on the data achieved through computational algorithms. 
The computational results are given extensively to classify the curves over integer modulo rings. 
\section{Summary of classification results over $\mathbb{F}_q$ and $\mathbb{Z}_n$}
To emphasize the structural parallelism between the classification of Weierstrass
elliptic curves over finite fields and finite rings, we summarize the results obtained in Sections \ref{cla} and \ref{clafq} in Table~\ref{tab:summary}.
The comparison highlights a symmetry between the two settings, where
the role of cardinality over finite fields is naturally replaced by Euler’s
totient function in the ring setting.

\begin{table}[ht]
\begin{minipage}{1\textwidth}
\centering
\resizebox{\textwidth}{!}{%
\renewcommand{\arraystretch}{1.25}
\begin{tabular}{|c|c|c|c|c|}
\hline
\textbf{Base} & \textbf{Model} & \textbf{Non-singular Curves} &
\textbf{Isomorphic to $E_k$} & \textbf{Isomorphism Classes} \\
\hline
$\mathbb{F}_q$ &
Reduced &
$N_r(\mathbb{F}_q) = q^2 - q$ &
$N^{(k)}_r(\mathbb{F}_q) = \dfrac{q-1}{|\mathrm{Aut}(E)|}$ &
$
C_r(\mathbb{F}_q)=
\begin{cases}
2q+6, & q \equiv 1 \pmod{12} \\
2q+2, & q \equiv 5 \pmod{12} \\
2q+4, & q \equiv 7 \pmod{12} \\
2q,   & q \equiv 11 \pmod{12}
\end{cases}
$ \\
\hline
$\mathbb{Z}_{n}$ &
Reduced &
$N_r(\mathbb{Z}_{n}) = \phi(n^{2})$ &
$N^{(k)}_r(\mathbb{Z}_n) = \dfrac{\phi(n)}{|\mathrm{Aut}(E)|}$ &
$
C_r(\mathbb{Z}_n) = \prod\limits_{i=1}^{l} C_r(\mathbb{Z}_{p_i^{e_i}}),\,\,
C_r(\mathbb{Z}_{p^m})=
\begin{cases}
2p^m+6, & p \equiv 1 \pmod{12} \\
2p^m+2, & p \equiv 5 \pmod{12} \\
2p^m+4, & p \equiv 7 \pmod{12} \\
2p^m,   & p \equiv 11 \pmod{12}
\end{cases}
$ \\
\hline
$\mathbb{F}_q$ &
Generalized &
$N_g(\mathbb{F}_q) = q^5 - q^4$ &
$N^{(k)}_g(\mathbb{F}_q) = \dfrac{q^4 - q^3}{|\mathrm{Aut}(E)|}$ &
$C_g(\mathbb{F}_q) = 2q + 3 + \left(\frac{-4}{q}\right) + 2\left(\frac{-3}{q}\right)$ \\
\hline
$\mathbb{Z}_n$ &
Generalized &
$N_g(\mathbb{Z}_n) = \phi(n^5)$ &
$N^{(k)}_g(\mathbb{Z}_n) = \dfrac{\phi(n^4)}{|\mathrm{Aut}(E)|}$ &
Open problem for $\mathbb{Z}_{2^m}$ and $\mathbb{Z}_{3^m}$ \\
\hline
\end{tabular}%
}
\caption{Comparison of classification results for reduced and generalized
Weierstrass elliptic curves over $\mathbb{F}_q$ and $\mathbb{Z}_n$.}
\label{tab:summary}
\end{minipage}
\end{table}
\section{Computational results}
\label{com}
This section presents computational results of a simple brute force algorithm implemented in C++ for classifying the elliptic curves over the ring $\mathbb{Z}_n$.
The crucial factor that contributes here is the isomorphism between two elliptic curves by taking the appropriate transformation map. 
Using this idea, we can fix the class leader or the primary non-singular elliptic curve that is isomorphic to the other candidates in the class through the respective transformation map.
Continuing the process over all the non-singular curves, we obtain the complete classification of elliptic curves over the finite ring $\mathbb{Z}_n$. We implemented the classification algorithm and carried out extensive computational experiments. Due to practical computational constraints, the exhaustive classification was performed for $\mathbb{Z}_n$ up to $n = 30$ in the generalized Weierstrass case and up to $n = 199$ in the reduced Weierstrass case. The corresponding codes and datasets are available at \cite{ECCguptalab}. This follows the subsequent results.
 \subsection{Reduced Weierstrass Elliptic Curves over $\mathbb{Z}_n$}
For the reduced Weierstrass elliptic curves over the finite rings $\mathbb{Z}_n$ where $6 \nmid n$, the computational data is comprehended in Table \ref{NSred} and Table \ref{ICred}.
 
Specifically, Table \ref{NSred} presents the values of $N_r(\mathbb{Z}_n)$ and results in Theorem \ref{rwt1}, Theorem \ref{nrwr} and Theorem \ref{nrc}.
\begin{table}[ht]
\centering

\begin{minipage}{0.48\textwidth}
\centering
\begin{tabular}{|c|c|}
\hline
$\mathbb{Z}_n$ & $N_r(\mathbb{Z}_n)$  \\ \hline
5  & 20  \\ \hline
7  & 42  \\ \hline
11 & 110 \\ \hline
13 & 156 \\ \hline
17 & 272 \\ \hline
19 & 342 \\ \hline
23 & 506 \\ \hline
25 & 500 \\ \hline
29 & 812 \\ \hline
31 & 930 \\ \hline
35 & 840 \\ \hline
37 & 1332 \\ \hline
41 & 1640 \\ \hline
43 & 1806 \\ \hline
47 & 2162 \\ \hline
49 & 2058 \\ \hline
53 & 2756 \\ \hline
55 & 2200 \\ \hline
59 & 3422 \\ \hline
\end{tabular}
\caption{Number of non-singular reduced Weierstrass elliptic curves $N_r(\mathbb{Z}_n)$ over $\mathbb{Z}_n$, $6 \nmid n$.}
\label{NSred}
\end{minipage}
\hfill
\begin{minipage}{0.48\textwidth}
\centering
\begin{tabular}{|c|c|}
\hline
$\mathbb{Z}_n$ & $C_r(\mathbb{Z}_n)$ \\ \hline
5  & 12  \\ \hline
7  & 18  \\ \hline
11 & 22  \\ \hline
13 & 32  \\ \hline
17 & 36  \\ \hline
19 & 42  \\ \hline
23 & 46  \\ \hline
25 & 52  \\ \hline
29 & 60  \\ \hline
31 & 66  \\ \hline
35 & 216 \\ \hline
37 & 80  \\ \hline
41 & 84  \\ \hline
43 & 90  \\ \hline
47 & 94  \\ \hline
49 & 102 \\ \hline
53 & 108 \\ \hline
55 & 264 \\ \hline
59 & 118 \\ \hline
\end{tabular}
\caption{Number of isomorphism classes $C_r(\mathbb{Z}_n)$ for reduced Weierstrass elliptic curves over $\mathbb{Z}_n$, $6 \nmid n$.}
\label{ICred}
\end{minipage}

\end{table}



For the computational data on the number of classes of reduced Weierstrass curves $C_r(\mathbb{Z}_n)$, one can look into Table \ref{ICred}, 
Also, the computational data validates Theorem \ref{irwr} and Theorem \ref{irw}.

\subsection{Generalized Weierstrass Elliptic Curves over $\mathbb{Z}_n$}
 
 We summarize the computational data related to the generalized Weierstrass elliptic curve in Table \ref{NSgen} and Table \ref{ICgen}.
 
 In particular, Table \ref{NSgen} presents the values of $N_g(\mathbb{Z}_n)$, coinciding with Theorem \ref{gwt1} and Theorem \ref{gwr}, Theorem \ref{ngc2} in the compact form over the finite fields and finite rings, respectively. 
 \begin{table}[ht]
\centering

\begin{minipage}{0.48\textwidth}
\centering
\begin{tabular}{|c|c|}
\hline
$\mathbb{Z}_n$ & $N_g(\mathbb{Z}_n)$ \\ \hline
2  & 16       \\ \hline
3  & 162      \\ \hline
4  & 512      \\ \hline
5  & 2500     \\ \hline
6  & 2592     \\ \hline
7  & 14406    \\ \hline
8  & 16384    \\ \hline
9  & 39366    \\ \hline
10 & 40000    \\ \hline
11 & 146410   \\ \hline
12 & 82944    \\ \hline
13 & 342732   \\ \hline
14 & 230496   \\ \hline
15 & 405000   \\ \hline
16 & 524288   \\ \hline
17 & 1336336  \\ \hline
18 & 629856   \\ \hline
19 & 2345778  \\ \hline
20 & 1280000  \\ \hline
\end{tabular}
\caption{Number of non-singular generalized Weierstrass elliptic curves $N_g(\mathbb{Z}_n)$ over $\mathbb{Z}_n$.}
\label{NSgen}
\end{minipage}
\hfill
\begin{minipage}{0.48\textwidth}
\centering
\begin{tabular}{|c|c|}
\hline
$\mathbb{Z}_n$ & $C_g(\mathbb{Z}_n)$ \\ \hline
2  & 5  \\ \hline
3  & 8  \\ \hline
4  & 8  \\ \hline
5  & 12 \\ \hline
6  & 40 \\ \hline
7  & 18 \\ \hline
8  & 16 \\ \hline
9  & 18 \\ \hline
10 & 60 \\ \hline
11 & 22 \\ \hline
12 & 64 \\ \hline
13 & 32 \\ \hline
14 & 90 \\ \hline
15 & 96 \\ \hline
16 & 32 \\ \hline
17 & 36 \\ \hline
18 & 90 \\ \hline
19 & 42 \\ \hline
20 & 96 \\ \hline
\end{tabular}
\caption{Number of isomorphism classes $C_g(\mathbb{Z}_n)$ for generalized Weierstrass elliptic curves over $\mathbb{Z}_n$.}
\label{ICgen}
\end{minipage}

\end{table}

In Table \ref{ICgen}, we present the computational results on the number of classes of generalized Weierstrass curves $C_g(\mathbb{Z}_n)$, that corresponds to  Theorem \ref{igw} and Theorem \ref{igwr}. 

To provide an extensive outlook on computational data, we classify the non-singular generalized Weierstrass elliptic curves with respective class leaders and appropriate transformation map over the ring $\mathbb{Z}_5$ in Table \ref{ICz4}.  
There are $12$ isomorphism classes with respective class leaders, each with No. of curves isomorphic to the corresponding class leader with a particular coordinate-based transformation function. 
All classes have different $|Aut(E)|$, which shows that when applying the transformation on each curve, precisely $|Aut(E)|$ transformations result in the same curve.

\begin{table}[]
\centering
\begin{tabular}{|c|c|c|c|c|}
\hline
Class & No. of & Class leader & $|Aut(E)|$ & No. of \\ 
& curves & & & points \\
\hline
1  & 2 & $y^2 = x^3 + 1$      & 2 & 5 \\ \hline
2  & 1 & $y^2 = x^3 + x$      & 4 & 3 \\ \hline
3  & 2 & $y^2 = x^3 + x + 1$  & 2 & 8 \\ \hline
4  & 2 & $y^2 = x^3 + x + 2$  & 2 & 3 \\ \hline
5  & 2 & $y^2 = x^3 + 2$      & 2 & 5 \\ \hline
6  & 1 & $y^2 = x^3 + 2x$     & 4 & 1 \\ \hline
7  & 2 & $y^2 = x^3 + 2x + 1$ & 2 & 6 \\ \hline
8  & 1 & $y^2 = x^3 + 3x$     & 4 & 9 \\ \hline
9  & 2 & $y^2 = x^3 + 3x + 2$ & 2 & 4 \\ \hline
10 & 1 & $y^2 = x^3 + 4x$     & 4 & 7 \\ \hline
11 & 2 & $y^2 = x^3 + 4x + 1$ & 2 & 7 \\ \hline
12 & 2 & $y^2 = x^3 + 4x + 2$ & 2 & 2 \\ \hline
\end{tabular}
\caption{Complete classification of the isomorphism classes of the generalized Weierstrass elliptic curves over $\mathbb{Z}_5$ (Note that, the number of points excludes point at infinity $\mathcal{O}$).}
\label{ICz4}
\end{table}
For the complete classification data in detail over finite rings $\mathbb{Z}_n$, readers are referred to \cite{ECCguptalab}. 
\section{Conclusion}
\label{con}
The present work gives the fundamental classification of Weierstrass elliptic curves over the finite ring $\mathbb{Z}_n$.
Both the generalized and reduced curves are considered in the study.
The precise expressions for enumerating non-singular elliptic curves, determining the count of curves isomorphic to a specific curve, and calculating the number of isomorphism classes within a finite ring are also provided. The extensive computational data\cite{ECCguptalab} is also included in this work to validate the results. Open problem is given. It would be interesting to establish Hasse-Weil-like bound on the number of points of an elliptic curve of $\mathbb{Z}_n$.  
\bibliographystyle{amsplain}
\bibliography{elliptic}

@book{menezes1993elliptic,
  title={Elliptic curve public key cryptosystems},
  author={Menezes, Alfred J},
  volume={234},
  year={1993},
  publisher={Springer Science \& Business Media}
}

@book{silverman2009arithmetic,
  title={The arithmetic of elliptic curves},
  author={Silverman, Joseph H},
  volume={106},
  year={2009},
  publisher={Springer}
}

@book{washington2008elliptic,
  title={Elliptic curves: number theory and cryptography},
  author={Washington, Lawrence C},
  year={2008},
  publisher={Chapman and Hall/CRC}
}

@book{lenstra1986elliptic,
  title={Elliptic curves and number-theoretic algorithms},
  author={Lenstra, Hendrik Willem and others},
  year={1986},
  publisher={Universiteit van Amsterdam Mathematisch Instituut}
}

@article{schmaleelliptic,
  title={Elliptic {C}urves over {R}ings with a {P}oint of {V}iew on {C}ryptography and {F}actoring},
  author={Schmale, Wiland and Quebbemann, Heinz-Georg and und Betreuung, Themenstellung and Rosenthal, Joachim}
}

@article{jeong2009isomorphism,
  title={Isomorphism classes of elliptic curves over finite fields with characteristic $3$},
  author={Jeong, Eunkyung},
  journal={Journal of the Chungcheong Mathematical Society},
  volume={22},
  number={3},
  pages={299--307},
  year={2009},
  publisher={The Chungcheong Mathematical Society}
}

@article{schoof1987nonsingular,
  title={Nonsingular plane cubic curves over finite fields},
  author={Schoof, Ren{\'e}},
  journal={Journal of combinatorial theory, Series A},
  volume={46},
  number={2},
  pages={183--211},
  year={1987},
  publisher={Elsevier}
}

@inproceedings{farashahi2011number,
  title={On the {N}umber of {D}istinct {L}egendre, {J}acobi, {H}essian and {E}dwards {C}urves},
  author={Farashahi, Reza Rezaeian},
  booktitle={WCC 2011-Workshop on coding and cryptography},
  pages={37--46},
  year={2011}
}

@article{rezaeian2010number,
  title={On the number of distinct elliptic curves in some families},
  author={Rezaeian Farashahi, Reza and Shparlinski, Igor E},
  journal={Designs, Codes and Cryptography},
  volume={54},
  pages={83--99},
  year={2010},
  publisher={Springer}
}

@article{feng2010elliptic,
  title={Elliptic curves in Huff’s model},
  author={Wu, Hongfeng and Feng, Rongquan},
  journal={Wuhan University Journal of Natural Sciences},
  volume={17},
  number={6},
  pages={473--480},
  year={2012},
  publisher={Springer}
}

@article{wu2011isomorphism,
  title={On the isomorphism classes of {L}egendre elliptic curves over finite fields},
  author={Wu, HongFeng and Feng, RongQuan},
  journal={Science in China A: Mathematics},
  volume={54},
  number={9},
  pages={1885--1890},
  year={2011}
}

@article{choie2004isomorphism,
  title={Isomorphism classes of elliptic and hyperelliptic curves over finite fields $\mathbb{F}_{(2g+ 1)^n}$},
  author={Choie, Youngju and Jeong, Eunkyung},
  journal={Finite Fields and Their Applications},
  volume={10},
  number={4},
  pages={583--614},
  year={2004},
  publisher={Elsevier}
}

@inproceedings{choie2002isomorphism,
  title={Isomorphism {C}lasses of {H}yperelliptic {C}urves of {G}enus $2$ over $\mathbb{F}_q$},
  author={Choie, Y and Yun, D},
  booktitle={ACISP},
  volume={2},
  pages={190--202},
  year={2002},
  organization={Springer}
}

@article{deng2006isomorphism,
  title={Isomorphism classes of hyperelliptic curves of genus $3$ over finite fields},
  author={Deng, Yingpu},
  journal={Finite Fields and Their Applications},
  volume={12},
  number={2},
  pages={248--282},
  year={2006},
  publisher={Elsevier}
}

@article{encinas2002isomorphism,
  title={Isomorphism classes of genus-$2$ hyperelliptic curves over finite fields},
  author={Encinas, L Hern{\'a}ndez and Menezes, Alfred J and Masqu{\'e}, J Munoz},
  journal={Applicable Algebra in Engineering, Communication and Computing},
  volume={13},
  pages={57--65},
  year={2002},
  publisher={Springer-Verlag}
}

@article{koblitz1987elliptic,
  title={Elliptic curve cryptosystems},
  author={Koblitz, Neal},
  journal={Mathematics of Computation},
  volume={48},
  number={177},
  pages={203--209},
  year={1987},
  publisher={American Mathematical Society}
}

@inproceedings{miller1985use,
  title={Use of elliptic curves in cryptography},
  author={Miller, Victor S},
  booktitle={Advances in Cryptology—CRYPTO’85},
  pages={417--426},
  year={1985},
  organization={Springer}
}

@article{SalaTaufer+2024,
url = {https://doi.org/10.1515/jmc-2023-0025},
title = {Group structure of elliptic curves over $\mathbb{Z}/{N}\mathbb{Z}$},
author = {Massimiliano Sala and Daniele Taufer},
pages = {20230025},
volume = {18},
number = {1},
journal = {Journal of Mathematical Cryptology},
doi = {doi:10.1515/jmc-2023-0025},
year = {2024},
lastchecked = {2024-02-23}
}

@inproceedings{waterhouse1969abelian,
  title={Abelian varieties over finite fields},
  author={Waterhouse, William C},
  booktitle={Annales scientifiques de l'{\'E}cole normale sup{\'e}rieure},
  volume={2},
  number={4},
  pages={521--560},
  year={1969}
}

@book{10.5555/1388394,
author = {Washington, Lawrence C.},
title = {Elliptic Curves: Number Theory and Cryptography},
year = {2008},
isbn = {9781420071467},
publisher = {Chapman \& Hall/CRC},
edition = {2}
}

@misc{ECCguptalab,
  author = {Parekh, Param and Parekh, Paavan and Deb, Sourav and Gupta, Manish K},
  howpublished = "site: \url{https://www.guptalab.org/ecc/}, codes: \url{https://github.com/guptalab/EC_Zn}",
  year = {2023}
}

@article{stangl1996counting,
  title={Counting {S}quares in $\mathbb{Z}_n$},
  author={Stangl, Walter D},
  journal={Mathematics Magazine},
  volume={69},
  number={4},
  pages={285--289},
  year={1996},
  publisher={Taylor \& Francis}
}

@article{serajcounting,
  title={Counting general power residues},
  author={Seraj, Samer},
  journal={Notes on Number Theory and Discrete Mathematics},
  volume={28},
  number={4},
  pages={730--743},
  doi={10.7546/nntdm.2022.28.4.730-743},
  year={2022}
}

@article{A.Chillali,
author = {Boulbot, A. and Chillali, Abdelhakim and Mouhib, A.},
year = {2019},
month = {02},
pages = {193},
title = {Elliptic {C}urves {O}ver the {R}ing ${R}$},
volume = {38},
journal = {Boletim da Sociedade Paranaense de Matemática},
doi = {10.5269/bspm.v38i3.39868}
}

@article{Hassib,
author = {Hassib, Moulay Hachem and Chillali, Abdelhakim and Elomary, Mohamed Abdou},
year = {2015},
month = {03},
pages = {},
title = {Elliptic curves over a chain ring of characteristic 3 (International Workshop of Algebra and Applications, 2014, FST Fez, Morocco)},
volume = {9},
journal = {Journal of Taibah University for Science},
doi = {10.1016/j.jtusci.2015.02.001}
}

@article{tadmori2015elliptic,
author = {Tadmori, Abdelhamid and Chillali, Abdelhakim and Ziane, M'Hammed},
year = {2015},
month = {01},
pages = {1721-1733},
title = {{Elliptic curve over ring $A_4=F_{2^d}[\epsilon]; \epsilon^4=0$}},
volume = {9},
journal = {Applied Mathematical Sciences},
doi = {10.12988/ams.2015.5147}
}

@article{Tadmori2015CryptographyOT,
  title={{Cryptography over the elliptic curve $E_{a,b}(A_3)$}},
  author={Abdelhamid Tadmori and Abdelhakim Chillali and M'hammed Ziane},
  journal={Journal of Taibah University for Science},
  year={2015},
  volume={9},
  pages={326 - 331},
  url={https://api.semanticscholar.org/CorpusID:121746255}
}

@article{Boulbot2016EllipticCO,
author = {Boulbot, Aziz and Chillali, Abdelhakim and Mouhib, Ali},
year = {2016},
month = {12},
pages = {123–129},
title = {Elliptic curves over the ring $\mathbb{F}_q[e], e^3 = e^2$},
volume = {4},
number = {4},
journal = {Gulf Journal of Mathematics},
doi = {10.56947/gjom.v4i4.271}
}

@article{euler_theoremata_1750,
	title = {Theoremata circa divisores numerorum},
	volume = {1},
	journal = {Novi Commentarii academiae scientiarum Petropolitanae},
	author = {Euler, Leonhard},
	year = {1750},
	pages = {20--48}
}

@article{euler_theoremata_1761,
	title = {Theoremata circa residua ex divisione potestatum relicta},
	volume = {7},
	journal = {Novi Commentarii academiae scientiarum Petropolitanae},
	author = {Euler, Leonhard},
	year = {1761},
	pages = {49--82}
}

@article{euler_criteria,
	title = {A generalization of Euler’s Criterion to composite moduli},
	volume = {7},
	journal = {Notes on Number Theory and Discrete Mathematics},
	author = {Vass, J.},
	year = {2016},
	pages = {9-19}
}

@article{MatharConj,
    author = {R. J. Mathar},
    title = {Size of the set of residues of integer powers of fixed exponent},
    note = {OEIS seq no. : A293482, accessed on January 13, 2026},
   
    journal = {OEIS},
    year = {2017},
    url = {https://oeis.org/A293482/a293482.pdf}
}

@book{rosen2011elementary,
  title={Elementary number theory},
  author={Rosen, Kenneth H},
  year={2011},
  publisher={Pearson Education London}
}

@article{kayal2006complexity,
  title={Complexity of ring morphism problems},
  author={Kayal, Neeraj and Saxena, Nitin},
  journal={computational complexity},
  volume={15},
  number={4},
  pages={342--390},
  year={2006},
  publisher={Springer}
}

\end{document}